\documentclass[english]{revtex4-1}
\usepackage[T1]{fontenc}
\usepackage[latin9]{inputenc}
\usepackage{geometry}
\usepackage{amsmath}
\usepackage{amsthm}
\usepackage{amssymb}
\usepackage{esint}

\makeatletter
\numberwithin{equation}{section}
\numberwithin{figure}{section}
\theoremstyle{plain}
\newtheorem{thm}{\protect\theoremname}
  \theoremstyle{plain}
  \newtheorem{lem}[thm]{\protect\lemmaname}
  \theoremstyle{plain}
  \newtheorem{prop}[thm]{\protect\propositionname}
  \theoremstyle{definition}
  \newtheorem{example}[thm]{\protect\examplename}

\makeatother

\usepackage{babel}
  \providecommand{\examplename}{Example}
  \providecommand{\lemmaname}{Lemma}
  \providecommand{\propositionname}{Proposition}
\providecommand{\theoremname}{Theorem}

\begin{document}

\title{Euler Polynomials and Identities for Non-Commutative Operators }

\author{Valerio De Angelis and Christophe Vignat }

\address{Department of Mathematics,  Xavier University of Louisiana, New Orleans, 
Department of Mathematics, Tulane University, New Orleans, \\ L.S.S. Supelec Universit\'{e} Paris Sud Orsay}
\begin{abstract}
Three kinds of identities involving non-commutating operators and
Euler and Bernoulli polynomials are studied. The first identity, as
given in \cite{Bender} by Bender and Bettencourt, expresses the nested
commutator of the Hamiltonian and momentum operators as the commutator
of the momentum and the shifted Euler polynomial of the Hamiltonian.
The second one, due to J.-C. Pain \cite{Pain}, links the commutators
and anti-commutators of the monomials of the position and momentum
operators. The third appears in a work by Figuieira de Morisson and
Fring \cite{Figueira} in the context of non-Hermitian Hamiltonian
systems. In each case, we provide several proofs and extensions of
these identities that highlight the role of Euler and Bernoulli polynomials.
\end{abstract}
\maketitle

\section{Introduction}

Special functions appear as a natural tool in many areas of theoretical
physics; in quantum physics, Hermite polynomials are the natural basis
to describe the behavior of the quantum harmonic oscillator. The Euler
and Bernoulli polynomials are ubiquitous in number theory and combinatorics;
in physics, they appear in various ways, such as in their association
with zeta functions and in their rich interplay with particle physics;
see \cite{Elizalde} for examples such as the Casimir effect and string
theory.

Bernoulli and Euler polynomials also appear in the field of non-commutative
operators in quantum physics, in more subtle ways; for example, the
Bernoulli numbers appear in the  the linear term of the celebrated
Baker-Campbell-Hausdorff formula. Their role in quantum algebras is
also detailed in \cite{Hodges}. This paper describes three distinct
contexts of quantum physics in which the Bernoulli and Euler polynomials
play an important role.

Define the \textit{commutator} of two operators $p$ and $q$ as
\[
\left[p,q\right]=pq-qp
\]
and their \textit{anti-commutator} as
\[
\left\{ p,q\right\} =pq+qp.
\]

Define moreover the \textit{nested (or iterated) commutators}
\[
\left[p,q\right]_{2}=\left[\left[p,q\right],q\right],\thinspace\thinspace\left[p,q\right]_{3}=\left[\left[\left[p,q\right],q\right],q\right]=\left[\left[p,q\right]_{2},q\right]
\]
and more generally
\[
\left[p,q\right]_{n}=\left[\left[p,q\right]_{n-1},q\right],\thinspace\thinspace n\ge1
\]
and accordingly the \textit{nested anti-commutators} as
\[
\left\{ p,q\right\} _{2}=\left\{ \left\{ p,q\right\} ,q\right\} ,\thinspace\thinspace\left\{ p,q\right\} _{3}=\left\{ \left\{ \left\{ p,q\right\} ,q\right\} ,q\right\} =\left\{ \left\{ p,q\right\} _{2},q\right\} 
\]
and more generally
\[
\left\{ p,q\right\} _{n}=\left\{ \left\{ p,q\right\} _{n-1},q\right\} ,\thinspace\thinspace n\ge1.
\]
For example
\[
\left\{ p,q\right\} _{2}=\left\{ pq+qp,q\right\} =pq^{2}+2qpq+q^{2}p
\]
and
\[
\left[p,q\right]_{2}=\left[pq-qp,q\right]=pq^{2}-2qpq+q^{2}p.
\]
It can be checked by induction on $n$ that the general cases are
\[
\left[p,q\right]_{n}=\sum_{k=0}^{n}\binom{n}{k}\left(-1\right)^{k}q^{k}pq^{n-k}
\]
and 
\[
\left\{ p,q\right\} _{n}=\sum_{k=0}^{n}\binom{n}{k}q^{k}pq^{n-k};
\]
as a consequence, the exponential generating functions of these two
sequences are
\[
\sum_{n\ge0}\left[p,q\right]_{n}\frac{z^{n}}{n!}=e^{-zq}pe^{zq},\thinspace\thinspace\sum_{n\ge0}\left\{ p,q\right\} _{n}\frac{z^{n}}{n!}=e^{zq}pe^{zq}.
\]
Note that $\left[p,q\right]=\left[p,q\right]_{1}$ and $\left\{ p,q\right\} =\left\{ p,q\right\} _{1}.$

In the sequel, we will consider the more particular case where the
non commutative operators $p$ and $q$ satisfy the identity
\begin{equation}
\left[q,p\right]=qp-pq=\imath\label{eq:noncom}
\end{equation}
where $\imath^{2}=-1.$ A more rigorous notation would be
\[
qp-pq=\imath I
\]
where $I$ is the identity operator, but in the absence of ambiguity,
we will systematically use the simplified notation (\ref{eq:noncom}).
For example, take $q=\imath\frac{d}{dx}$ and $p=x$ in the sense
that $q$ and $p$ act on any differentiable function $f$ as follows
\[
qf\left(x\right)=\imath\frac{d}{dx}f\left(x\right)=\imath f'\left(x\right)
\]
and
\[
pf\left(x\right)=xf\left(x\right).
\]
Then $\left(qp-pq\right)f=\imath f$ so that indeed
\[
qp-pq=\imath,
\]
in the sense that, for any differentiable function $f,$
\[
\left(qp-pq\right)f=\imath f.
\]
However, we wish to prove the results stated in the next sections
in the general case of two arbitrary operators $p$ and $q$, assuming
only that $p$ and $q$ satisfy (\ref{eq:noncom}). 

Let us introduce moreover another operator, the \textit{Hamiltonian
operator} $H$, as
\begin{equation}
H=\frac{1}{2}\left(p^{2}+q^{2}\right).\label{eq:Hamiltonian}
\end{equation}

Our aim in this paper is to study three types of identities on specific
commutators and anti-commutators; the first one was studied by Bender
and Bettencourt \cite{Bender}, and links the $n-$nested anti-commutator
of the operators $q$ and $H$ to a simpler anti-commutator of $q$
and a polynomial version of $H.$ This result was not proved, but
only inferred by these authors. Only recently, a proof was provided
in \cite{Galapon}, based on the integral representation of operators,
more precisely the generalized Weyl transform. Our approach consists
in first transforming Bender and Bettencourt's identity into a simpler
form, and then in providing two different proofs, one algebraic and
another based on the properties of some operators. A third, purely
analytic proof, is obtained by realizing the considered operators
and looking at their action on a sufficiently large set of functions.

The second identity studied in this paper was introduced recently
by Pain \cite{Pain}: it expresses the commutator of monomials in
operators $p$ and $q$ as a linear combination of their anti-commutators.
We show that the use of generating functions not only gives a simple
proof of this result, but also allows to derive the converse identity,
expressing the anti-commutator of monomials of $p$ and $q$ as a
linear combination of their commutators. In both cases, the Euler
and Bernoulli polynomials play an important role.

The third identity appears in the context of non-Hermitian Hamiltonian
systems, as studied by Figuieira de Morisson and Fring in \cite{Figueira}:
the real and imaginary part of the non-Hermitian Hamiltonian can be
shown, under assumptions that will be detailed later, to be related
by a linear identity involving Euler numbers.

One of the main intents of this paper is to show that for a given
identity on non-commutating operators, a variety of different proofs
- each pertaining to a different area of mathematics - can be exhibited.
On the way of each of these proofs, some interesting results, of either
combinatorial, algebraic or analytical nature, may appear.

\section{An Identity by Bender and Bettencourt}

\subsection{Introduction}

In \cite{Bender}, C. Bender and L. Bettencourt inspect the first
values of the operator
\[
\frac{1}{2^{n}}\left\{ q,H\right\} _{n};
\]
for $n=0,1,2,3,4,5$ and $6,$ they find respectively
\[
1,\thinspace\thinspace\frac{1}{2}\left\{ q,H\right\} ,\thinspace\thinspace\frac{1}{2}\left\{ q,H^{2}-\frac{1}{4}\right\} ,\thinspace\thinspace\frac{1}{2}\left\{ q,H^{3}-\frac{3}{4}H\right\} 
\]
\[
\frac{1}{2}\left\{ q,H^{4}-\frac{3}{2}H^{2}+\frac{5}{16}\right\} ,\thinspace\thinspace\frac{1}{2}\left\{ q,H^{5}-\frac{5}{2}H^{3}+\frac{25}{16}H\right\} 
\]
and
\[
\frac{1}{2}\left\{ q,H^{6}-\frac{15}{4}H^{4}+\frac{75}{16}H^{2}-\frac{61}{64}\right\} .
\]
This suggests the following result 
\begin{equation}
\frac{1}{2^{n}}\left\{ q,H\right\} _{n}=\frac{1}{2}\left\{ q,E_{n}\left(H+\frac{1}{2}\right)\right\} \label{eq:Benrder1}
\end{equation}
where, by a real \textit{tour de force}, the authors identify $E_{n}\left(x\right)$
as the Euler polynomial of degree $n,$ defined by the exponential
generating function
\begin{equation}
\sum_{n\ge0}\frac{E_{n}\left(x\right)}{n!}z^{n}=\frac{2e^{zx}}{e^{z}+1},\thinspace\thinspace\thinspace\vert z\vert<2\pi.\label{eq:EulerGF}
\end{equation}
Since this identity is only inferred in \cite{Bender}, we propose
to prove it in several ways using algebraic or analytic methods.

\subsection{Preprocessing}

As a first step of the proof, we transform the identity (\ref{eq:Benrder1})
into the equivalent form
\begin{equation}
\frac{1}{2^{n}}\left\{ q,H-\frac{1}{2}\right\} _{n}=\frac{1}{2}\left\{ q,E_{n}\left(H\right)\right\} .\label{eq:equivalent}
\end{equation}

Next, we consider the following lemma.
\begin{lem}
\label{lem:lemma 1}A polynomial $P_{n}\left(x\right)$ satisfies
\begin{equation}
\frac{1}{2}P_{n}\left(x\right)+\frac{1}{2}P_{n}\left(x+1\right)=x^{n}\label{eq:Euler}
\end{equation}
if and only if
\[
P_{n}\left(x\right)=E_{n}\left(x\right),
\]
the Euler polynomial of degree $n.$
\begin{proof}
The fact that the Euler polynomials satisfy 
\begin{equation}
E_{n}\left(x\right)+E_{n}\left(x+1\right)=2x^{n},\thinspace\thinspace\thinspace n\in\mathbb{N},\label{eq:Eulermean}
\end{equation}
can be deduced from the generating function (\ref{eq:EulerGF}). The
fact that this is an equivalence can be checked as follows: starting
from
\[
E_{n}\left(x\right)+E_{n}\left(x+1\right)=P_{n}\left(x\right)+P_{n}\left(x+1\right)
\]
which shows that $P_{n}\left(x\right)$ is of degree $n$, and writing
\[
P_{n}\left(x\right)=a_{n}^{\left(n\right)}x^{n}+\dots
\]
then, taking $x\to+\infty$ shows that $a_{n}^{\left(n\right)}$ is
also the leading term of $E_{n}\left(x\right).$ Then applying the
same argument to the new polynomial
\[
P_{n}\left(x\right)-a_{n}^{\left(n\right)}x^{n}
\]
shows by induction on $n$ that both polynomials $P_{n}$ and $E_{n}$
coincide.

Another proof of this Lemma uses the integral representation of Euler
polynomials that can be found as formula 24.7.6 in \cite{Nist},
\begin{equation}
E_{n}\left(z\right)=\int_{-\infty}^{+\infty}\left(\imath u-\frac{1}{2}+z\right)^{n}sech\left(\pi u\right)du.\label{eq:Euler integral}
\end{equation}
Start expanding
\[
P_{n}\left(x\right)=\sum_{k=0}^{n}p_{k}^{\left(n\right)}x^{k}
\]
and rewrite (\ref{eq:Euler}) as
\[
\frac{1}{2}\sum_{k=0}^{n}p_{k}^{\left(n\right)}x^{k}+\frac{1}{2}\sum_{k=0}^{n}p_{k}^{\left(n\right)}\left(x+1\right)^{k}=x^{n}.
\]
Substitute $x\to\imath u-\frac{1}{2}+x$ in both sides and integrate
again the function $sech\left(\pi u\right),$ giving
\[
\frac{1}{2}\sum_{k=0}^{n}p_{k}^{\left(n\right)}E_{k}\left(x\right)+\frac{1}{2}\sum_{k=0}^{n}p_{k}^{\left(n\right)}E_{k}\left(x+1\right)=E_{n}\left(x\right).
\]
Using the property (\ref{eq:Eulermean}) of Euler polynomials, deduce
\[
\sum_{k=0}^{n}p_{k}^{\left(n\right)}x^{k}=P_{n}\left(x\right)=E_{n}\left(x\right)
\]

\end{proof}
\end{lem}
Let us introduce next the symbolic notation, with $a\in\mathbb{R},$
\begin{equation}
\left(\left\{ q,H\right\} +a\right)_{n}=\sum_{k=0}^{n}\binom{n}{k}a^{n-k}\left\{ q,H\right\} _{k}\label{eq:sub_n}
\end{equation}
and the convention $\left\{ q,H\right\} _{0}=q.$

We will need the following identity, the proof of which is elementary.
\begin{lem}
\label{lem:Lemma 2}For $a\in\mathbb{R}$ and $n\in\mathbb{N},$

\[
\left\{ q,H+\frac{a}{2}\right\} _{n}=\left(\left\{ q,H\right\} +a\right)_{n}.
\]
\end{lem}
\begin{proof}
Take the exponential generating function on both sides. On the left-hand
side
\[
\sum_{n\ge0}\left\{ q,H+\frac{a}{2}\right\} _{n}\frac{z^{n}}{n!}=e^{z\left(H+\frac{a}{2}\right)}qe^{z\left(H+\frac{a}{2}\right)}
\]
while on the right-hand side
\[
\sum_{n\ge0}\left(\left\{ q,H\right\} +a\right)_{n}\frac{z^{n}}{n!}=e^{za}e^{zH}qe^{zH}
\]
so that both generating functions coincide. Since they are analytic
functions, the corresponding sequences coincide.
\end{proof}
Using both previous lemmas, we can now deduce the following equivalence.
\begin{prop}
An equivalent form of identity (\ref{eq:Benrder1}) (or (\ref{eq:equivalent}))
is
\[
\frac{1}{2^{n}}\left\{ q,H-\frac{1}{2}\right\} _{n}+\frac{1}{2^{n}}\left\{ q,H+\frac{1}{2}\right\} _{n}=\left\{ q,H^{n}\right\} .
\]
\end{prop}
\begin{proof}
Assume first that
\[
\frac{1}{2^{n}}\left\{ q,H-\frac{1}{2}\right\} _{n}=\frac{1}{2}\left\{ q,E_{n}\left(H\right)\right\} 
\]
holds. Then
\[
\frac{1}{2^{n}}\left\{ q,H-\frac{1}{2}\right\} _{n}+\frac{1}{2^{n}}\left\{ q,H+\frac{1}{2}\right\} _{n}=\left\{ q,\frac{1}{2}\left(E_{n}\left(H\right)+E_{n}\left(H+1\right)\right)\right\} =\left\{ q,H^{n}\right\} 
\]
where we have used the property (\ref{eq:Eulermean}) of Euler polynomials.

Reciprocally, assume that
\[
\frac{1}{2^{n}}\left\{ q,H-\frac{1}{2}\right\} _{n}+\frac{1}{2^{n}}\left\{ q,H+\frac{1}{2}\right\} _{n}=\left\{ q,H^{n}\right\} 
\]
holds. As can be checked from their generating function (\ref{eq:EulerGF}),
the Euler polynomials are Appell polynomials, i.e. they satisfy
\[
E_{n}\left(x\right)=\sum_{k=0}^{n}\binom{n}{k}E_{n-k}\left(0\right)x^{k};
\]
hence we have
\[
\left\{ q,E_{n}\left(H\right)\right\} =\left\{ q,\sum_{k=0}^{n}\binom{n}{k}E_{n-k}\left(0\right)H^{k}\right\} 
\]
so that we need to compute
\[
\sum_{k=0}^{n}\binom{n}{k}E_{n-k}\left(0\right)\frac{1}{2^{k}}\left\{ q,H-\frac{1}{2}\right\} _{k}\thinspace\thinspace\text{and}\thinspace\thinspace\sum_{k=0}^{n}\binom{n}{k}E_{n-k}\left(0\right)\frac{1}{2^{k}}\left\{ q,H+\frac{1}{2}\right\} _{k}
\]
and to show that their sum is equal to 
\[
\frac{1}{2^{n-1}}\left\{ q,H-\frac{1}{2}\right\} _{n}.
\]
By Lemma (\ref{lem:Lemma 2}), the first sum is
\[
\sum_{k=0}^{n}\binom{n}{k}E_{n-k}\left(0\right)\frac{1}{2^{k}}\left\{ q,H-\frac{1}{2}\right\} _{k}=\sum_{k=0}^{n}\binom{n}{k}E_{n-k}\left(0\right)\frac{1}{2^{k}}\left(\left\{ q,H\right\} -1\right)_{k}.
\]
Applying (\ref{eq:sub_n}), this is
\begin{eqnarray*}
\sum_{k=0}^{n}\binom{n}{k}E_{n-k}\left(0\right)\frac{1}{2^{k}}\sum_{l=0}^{k}\binom{k}{l}\left(-1\right)^{k-l}\left\{ q,H\right\} _{l} & = & \sum_{l=0}^{n}\binom{n}{l}\left\{ q,H\right\} _{l}\sum_{p=0}^{n-l}\binom{n-l}{p}E_{n-l-p}\left(0\right)\frac{\left(-1\right)^{p}}{2^{p+l}}\\
 & = & \sum_{l=0}^{n}\binom{n}{l}\frac{\left\{ q,H\right\} _{l}}{2^{l}}E_{n-k}\left(-\frac{1}{2}\right).
\end{eqnarray*}
Thus
\[
\sum_{k=0}^{n}\binom{n}{k}E_{n-k}\left(0\right)\frac{1}{2^{k}}\left\{ q,H-\frac{1}{2}\right\} _{k}=\sum_{l=0}^{n}\binom{n}{l}\frac{\left\{ q,H\right\} _{l}}{2^{l}}E_{n-l}\left(-\frac{1}{2}\right).
\]
and similarly
\[
\sum_{k=0}^{n}\binom{n}{k}E_{n-k}\left(0\right)\frac{1}{2^{k}}\left\{ q,H+\frac{1}{2}\right\} _{k}=\sum_{l=0}^{n}\binom{n}{l}\frac{\left\{ q,H\right\} _{l}}{2^{l}}E_{n-l}\left(\frac{1}{2}\right)
\]
so that the sum equals
\begin{eqnarray*}
\frac{1}{2^{n}}\sum_{l=0}^{n}\binom{n}{l}\frac{\left\{ q,H\right\} _{l}}{2^{l}}\left(E_{n-l}\left(-\frac{1}{2}\right)+E_{n-l}\left(\frac{1}{2}\right)\right) & = & \frac{1}{2^{n-1}}\sum_{l=0}^{n}\binom{n}{l}\frac{\left\{ q,H\right\} _{l}}{2^{l}}\left(-\frac{1}{2}\right)^{n-l}
\end{eqnarray*}
where we have used (\ref{eq:Eulermean}) again.

The last sum is now identified by Lemma (\ref{lem:Lemma 2}) as 
\[
\frac{2}{2^{n}}\left(\left\{ q,H\right\} -1\right)_{n}=\frac{1}{2^{n-1}}\left\{ q,H-\frac{1}{2}\right\} _{n},
\]
and this is the desired result.
\end{proof}
We conclude that proving (\ref{eq:Benrder1}) is equivalent to be
proving the following result.
\begin{thm}
\label{thm:Thm4}If $p$ and $q$ are two operators that satisfy (\ref{eq:noncom})
and $H$ is the associated Hamiltonian defined as in (\ref{eq:Hamiltonian}),
then for any integer $n\ge0,$
\begin{equation}
\frac{1}{2^{n}}\left(\left\{ q,H\right\} -1\right)_{n}+\frac{1}{2^{n}}\left(\left\{ q,H\right\} +1\right)_{n}=\left\{ q,H^{n}\right\} \label{eq:main identity}
\end{equation}
with the symbolic notation $\left(\left\{ q,H\right\} +a\right)_{n}$
as defined in (\ref{eq:sub_n}).\end{thm}
\begin{example}
The case $n=1$ reads
\[
\frac{1}{2}\left(\left\{ q,H\right\} -1\right)_{1}+\frac{1}{2}\left(\left\{ q,H\right\} +1\right)_{1}=\left\{ q,H\right\} .
\]
This is easily checked by expanding the left-hand side as
\[
\frac{1}{2}\left(\left\{ q,H\right\} -1\right)_{1}+\frac{1}{2}\left(\left\{ q,H\right\} +1\right)_{1}=\frac{1}{2}\left(\left\{ q,H\right\} -1\right)+\frac{1}{2}\left(\left\{ q,H\right\} +1\right)=\left\{ q,H\right\} .
\]

\begin{example}
The case $n=2$ reads
\[
\frac{1}{4}\left(\left\{ q,H\right\} -1\right)_{2}+\frac{1}{4}\left(\left\{ q,H\right\} +1\right)_{2}=\left\{ q,H^{2}\right\} .
\]
The left-hand side is now
\[
\frac{1}{4}\left(\left\{ q,H\right\} -1\right)_{2}+\frac{1}{4}\left(\left\{ q,H\right\} +1\right)_{2}=\frac{1}{2}\left\{ q,H\right\} _{2}+\frac{1}{2}\left\{ q,H\right\} _{0}=\frac{1}{2}\left(\left\{ q,H\right\} _{2}+q\right).
\]
But since $\left[q,H\right]_{2}=\left[\left[q,H\right],H\right]$
and $\left[q,H\right]=\imath p,$ we deduce $\left[q,H\right]_{2}=\left[\imath p,H\right]=\imath\left[p,H\right].$
It then follows from $\left[p,H\right]=-\imath q$ that $\left[q,H\right]_{2}=q.$

We deduce
\begin{eqnarray*}
\frac{1}{2}\left(\left\{ q,H\right\} _{2}+q\right) & = & \frac{1}{2}\left(\left\{ q,H\right\} _{2}+\left[q,H\right]_{2}\right)\\
 & = & \frac{1}{2}\left(qH^{2}+2HqH+H^{2}q\right)+\frac{1}{2}\left(qH^{2}-2HqH+H^{2}q\right)\\
 & = & qH^{2}+H^{2}q=\left\{ q,H^{2}\right\} 
\end{eqnarray*}
which is the desired result.
\end{example}
\end{example}

\subsection{An algebraic proof }

The first proof of Thm \ref{thm:Thm4} (and thus of Bender's identity
(\ref{eq:Benrder1})) is given in the case $n$ even only, replacing
$n$ by $2n$ in (\ref{eq:main identity}): its extension to the odd
case follows the same lines. It starts with the expansion
\[
\frac{1}{2}\left(\left\{ q,H\right\} -1\right)_{2n}+\frac{1}{2}\left(\left\{ q,H\right\} +1\right)_{2n}=\sum_{k=0}^{n}\binom{2n}{2k}\left\{ q,H\right\} _{2k}
\]
in which $\left\{ q,H\right\} _{2k}$ is replaced by $\left\{ \left[q,H\right]_{2n-2k},H\right\} _{2k}$
since $\left[q,H\right]_{2n-2k}=q.$ We obtain
\begin{eqnarray*}
\sum_{k=0}^{n}\binom{2n}{2k}\left\{ \left[q,H\right]_{2n-2k},H\right\} _{2k}=\sum_{k=0}^{n}\binom{2n}{2k}\sum_{l=0}^{2k}\binom{2k}{l}H^{l}\left[q,H\right]_{2n-2k}H^{2k-l}\\
=\sum_{k=0}^{n}\binom{2n}{2k}\sum_{l=0}^{2k}\binom{2k}{l}H^{l}\left(\sum_{r=0}^{2n-2k}\binom{2n-2k}{r}\left(-1\right)^{r}H^{r}qH^{2n-2k-r}\right)H^{2k-l}\\
=\sum_{k=0}^{n}\sum_{l=0}^{2k}\sum_{r=0}^{2n-2k}\binom{2n}{2k}\binom{2k}{l}\binom{2n-2k}{r}\left(-1\right)^{r}H^{l+r}qH^{2n-r-l}\\
=\sum_{k=0}^{n}\sum_{l=0}^{2k}\sum_{s=r+l=0}^{2n}\binom{2n}{2k}\binom{2k}{l}\binom{2n-2k}{s-l}\left(-1\right)^{s-l}H^{s}qH^{2n-s}.
\end{eqnarray*}
We thus need to show that
\[
b_{s}=\sum_{k=0}^{n}\sum_{l=0}^{2k}\binom{2n}{2k}\binom{2k}{l}\binom{2n-2k}{s-l}\left(-1\right)^{s-l}=\begin{cases}
0 & 1\le s\le2n-1\\
2^{2n-1} & s=0,\thinspace\thinspace s=2n
\end{cases}
\]
Differentiating the identity
\[
\sum_{k=0}^{n}\binom{2n}{2k}x^{2k}y^{2n-2k}=\frac{1}{2}\left(x+y\right)^{2n}+\frac{1}{2}\left(x-y\right)^{2n}
\]
$i$ times with respect to $x$ and $j$ times with respect to $y,$
and then evaluating at $x=y=1$ gives
\[
\sum_{k=0}^{n}\binom{2n}{2k}\binom{2k}{i}\binom{2n-2k}{j}=\binom{2n}{i}\binom{2n-i}{j}2^{2n-i-j-1}
\]
if $0\le i\le n,\thinspace\thinspace0\le j\le n$ and $i+j\le2n-1,$
while
\[
\sum_{k=0}^{n}\binom{2n}{2k}\binom{2k}{n}\binom{2n-2k}{n}=\frac{1}{2}\binom{2n}{n}\left(1+\left(-1\right)^{n}\right)
\]
if $i=j=n.$ We deduce, for $1\le s\le2n-1,$
\begin{eqnarray*}
b_{s} & = & \sum_{l=0}^{s}\sum_{k=0}^{n}\binom{2n}{2k}\binom{2k}{l}\binom{2n-2k}{s-l}\left(-1\right)^{s-l}=\sum_{l=0}^{s}\left(-1\right)^{s-l}\binom{2n}{l}\binom{2n-l}{s-l}2^{2n-s-1}\\
 & = & 2^{2n-s-1}\binom{2n}{s}\sum_{l=0}^{s}\left(-1\right)^{s-l}\binom{s}{l}=0.
\end{eqnarray*}
A straightforward computation gives $b_{0}=b_{2n}=2^{2n-1}.$ This
concludes the proof.

\subsection{An operator-based proof}

Another version of the proof of Thm \ref{thm:Thm4} is obtained now by identifying the nested commutators and anti-commutators $\left[q,H\right]_{k}$ and $\left\{q,H\right\}_{k}$ as powers of simple operators.

Define the operators $A_H$ and $B_H$ acting on $q$ as
\[
A_{H}q=\left[q,H\right],\,\,\,\, B_{H}q=\left\{q,H\right\}.
\]
We first verify the following identities.
\begin{lem}
The operators $A_{H}$ and $B_{H}$ satisfy, for $k \in \mathbb{N},$
\begin{equation}
A_{H}^{k}q=\left[q,H\right]_{k},\thinspace\thinspace B_{H}^{k}q=\left\{q,H\right\}_{k}
\label{eq:shift1}
\end{equation}
and
\begin{equation}
\left(A_{H}+B_{H}\right)^{k}q=2^{k}qH^{k},\thinspace\thinspace  \left(A_{H}-B_{H}\right)^{k}q=2^{k}H^{k}q.
\label{eq:shift2}
\end{equation}
Moreover, $A_{H}$ and $B_{H}$ commute.
\end{lem}
\begin{proof}
The identities \eqref{eq:shift1} and \eqref{eq:shift2} are easily proved by induction on $k.$
Moreover, $B_{H}A_{H}q=qH^2-H^2q=A_{H}B_{H}q$ so that $A_{H}$ and $B_{H}$ commute.
\end{proof}
We obtain the final identity using (\ref{eq:shift1}) and (\ref{eq:shift2})
as follows:
\begin{eqnarray*}
\sum_{k=0}^{n}\binom{2n}{2k}\left\{ \left[q,H\right]_{2n-2k},H\right\} _{2k} & = & \sum_{k=0}^{n}\binom{2n}{2k}B_{H}^{2k}A_{H}^{2n-2k}q\\
 & = & \frac{1}{2}\left(\left(A_{H}+B_{H}\right)^{2n}+\left(A_{H}-B_{H}\right)^{2n}\right)q\\
 & = & 2^{2n-1}\left(qH^{2n}+H^{2n}q\right)=2^{2n-1}\left\{ q,H^{2n}\right\} .
\end{eqnarray*}

\subsection{Another Proof: an analytic approach}

The third proof of Thm \ref{thm:Thm4} is purely analytic: we choose
the realization
\[
p=x,\thinspace\thinspace q=\imath\frac{d}{dx}.
\]
This is a faithful realization \cite{Flajolet}: any result proved
with this realization holds for any other operators $p$ and $q$
that satisfy the commutation relation (\ref{eq:noncom}). 

Next we consider the Hermite functions defined by the Rodrigues formula
\[
\psi_{n}\left(x\right)=\left(-1\right)^{n}\frac{1}{\sqrt{2^{n}n!\sqrt{\pi}}}e^{\frac{x^{2}}{2}}\frac{d^{n}}{dx^{n}}e^{-x^{2}},\thinspace\thinspace n\in\mathbb{N}.
\]
These functions satisfy the differential equation
\[
-\psi"_{n}\left(x\right)+x^{2}\psi_{n}\left(x\right)=\left(2n+1\right)\psi_{n}\left(x\right)
\]
or equivalently
\begin{equation}
H\psi_{n}\left(x\right)=\left(n+\frac{1}{2}\right)\psi_{n}\left(x\right),\label{eq:HPsi}
\end{equation}
so that they are eigenfunctions of the Hamiltonian. Since they form
a basis of $L^{2}\left(\mathbb{R}\right),$ it is sufficient to verify
(\ref{eq:main identity}) on this set of functions. But it is easily
verified that
\[
\left\{ q,H\right\} _{n}\psi_{l}\left(x\right)=\sum_{k=0}^{n}\binom{n}{k}H^{k}qH^{n-k}\psi_{l}\left(x\right)
\]
with, from (\ref{eq:HPsi}),
\[
H^{n-k}\psi_{l}\left(x\right)=\left(l+\frac{1}{2}\right)^{n-k}\psi_{l}\left(x\right)
\]
and with
\[
q\psi_{l}\left(x\right)=\imath\left(\sqrt{\frac{l}{2}}\psi_{l-1}\left(x\right)-\sqrt{\frac{l+1}{2}}\psi_{l+1}\left(x\right)\right)
\]
so that
\begin{eqnarray*}
\left\{ q,H\right\} _{n}\psi_{l}\left(x\right) & = & \sum_{k=0}^{n}\binom{n}{k}H^{k}qH^{n-k}\psi_{l}\left(x\right)\\
=\sum_{k=0}^{n}\binom{n}{k}\left(l-\frac{1}{2}\right)^{k}\imath\sqrt{\frac{l}{2}}\left(l+\frac{1}{2}\right)^{n-k}\psi_{l-1}\left(x\right) & - & \sum_{k=0}^{n}\binom{n}{k}\left(l+\frac{3}{2}\right)^{k}\imath\sqrt{\frac{l+1}{2}}\left(l+\frac{1}{2}\right)^{n-k}\psi_{l+1}\left(x\right)\\
=\imath\sqrt{\frac{l}{2}}\left(2l\right)^{n}\psi_{l-1}\left(x\right)-\imath\sqrt{\frac{l+1}{2}}\left(2l+2\right)^{n}\psi_{l+1}\left(x\right) & = & \imath2^{n-\frac{1}{2}}\left(l^{n+\frac{1}{2}}\psi_{l-1}\left(x\right)-\left(l+1\right)^{n+\frac{1}{2}}\psi_{l+1}\left(x\right)\right).
\end{eqnarray*}
We deduce
\begin{eqnarray*}
\left(\left\{ q,H\right\} +1\right)_{n}\psi_{l}\left(x\right) & = & \sum_{k=0}^{n}\binom{n}{k}\left\{ q,H\right\} _{k}\psi_{l}\left(x\right)\\
 & = & \sum_{k=0}^{n}\binom{n}{k}\imath2^{k-\frac{1}{2}}\left(l^{k+\frac{1}{2}}\psi_{l-1}\left(x\right)-\left(l+1\right)^{k+\frac{1}{2}}\psi_{l+1}\left(x\right)\right)\\
 & = & \imath\sqrt{\frac{l}{2}}\left(1+2l\right)^{n}\psi_{l-1}\left(x\right)-\imath\sqrt{\frac{l+1}{2}}\left(2l+3\right)^{n}\psi_{l+1}\left(x\right)
\end{eqnarray*}
and accordingly
\[
\left(\left\{ q,H\right\} -1\right)_{n}\psi_{l}\left(x\right)=\imath\sqrt{\frac{l}{2}}\left(2l-1\right)^{n}\psi_{l-1}\left(x\right)-\imath\sqrt{\frac{l+1}{2}}\left(2l+1\right)^{n}\psi_{l+1}\left(x\right)
\]
so that
\begin{eqnarray*}
\left[\left(\left\{ q,H\right\} +1\right)_{n}+\left(\left\{ q,H\right\} -1\right)_{n}\right]\psi_{l}\left(x\right) & = & \imath\sqrt{\frac{l}{2}}\left[\left(2l+1\right)^{n}+\left(2l-1\right)^{n}\right]\psi_{l-1}\left(x\right)\\
 & - & \imath\sqrt{\frac{l+1}{2}}\left[\left(2l+3\right)^{n}+\left(2l+1\right)^{n}\right]\psi_{l+1}\left(x\right).
\end{eqnarray*}
Now
\begin{eqnarray*}
\left\{ q,H^{n}\right\} \psi_{l}\left(x\right) & = & qH^{n}\psi_{l}\left(x\right)+H^{n}q\psi_{l}\left(x\right)\\
 & = & \left(l+\frac{1}{2}\right)^{n}\imath\left[\sqrt{\frac{l}{2}}\psi_{l-1}\left(x\right)-\sqrt{\frac{l+1}{2}}\psi_{l+1}\left(x\right)\right]\\
 & + & \imath\sqrt{\frac{l}{2}}\left(l-\frac{1}{2}\right)^{n}\psi_{l-1}\left(x\right)-\imath\sqrt{\frac{l+1}{2}}\left(l+\frac{3}{2}\right)^{n}\psi_{l+1}\left(x\right)
\end{eqnarray*}
which proves the result.

\section{identities on commutators and anti-commutators of monomials}

\subsection{Introduction}

In this second part, we change notations compared to the first part:
in order to follow J.-C. Pain's notations, we assume now that
the operators $p$ and $q$ satisfy the more general commutation relation
\[
\left[p,q\right]=pq-qp=c,
\]
with $c\in\mathbb{C}\backslash\left\{ 0\right\} .$ Bender's results
in the former section correspond to $c=-\imath.$ 

In the recent publication \cite{Pain}, J.-C. Pain derived the following
identity between commutators and anti-commutators of monomials of
the operators $p$ and $q.$
\begin{thm}
\cite[eq. (43)]{Pain}If $p$ and $q$ are two operators such that
\[
\left[p,q\right]=pq-qp=c,
\]
then the commutators $\left[p^{n},q^{m}\right]$ and anti-commutators
$\left\{ p^{n},q^{m}\right\} $ of monomials are related by the convolution
identity 
\begin{equation}
\left[\frac{p^{n}}{n!},\frac{q^{m}}{m!}\right]=-\sum_{k=1}^{\min\left(m,n\right)}c^{k}\frac{E_{k}\left(0\right)}{k!}\left\{ \frac{p^{n-k}}{\left(n-k\right)!},\frac{q^{m-k}}{\left(m-k\right)!}\right\} \label{eq:pain}
\end{equation}
where $E_{k}\left(0\right)$ is the Euler polynomial \footnote{Note that $E_{k}\left(0\right)$ should not be confused with the $k-$th
Euler number $E_{k}$ defined as $E_{k}=2^{k}E_{k}\left(\frac{1}{2}\right)$} of degree $k$ evaluated at $0,$ as defined by (\ref{eq:EulerGF}).
\end{thm}
The proof given in \cite{Pain} relies on the fact that if the identity
\begin{equation}
\left[\frac{p^{n}}{n!},\frac{q^{m}}{m!}\right]=\sum_{k=1}^{\min\left(m,n\right)}c^{k}\frac{v_{k}}{k!}\left\{ \frac{p^{n-k}}{\left(n-k\right)!},\frac{q^{m-k}}{\left(m-k\right)!}\right\} \label{eq:pain2}
\end{equation}
holds for some sequence $\left\{ v_{k}\right\} ,$ then these numbers
$v_{k}$ should satisfy (see \cite{Pain})
\[
v_{k}+\sum_{l=1}^{k}\binom{k}{l}v_{l}=1,
\]
which, with $v_{0}=1,$ characterizes exactly these numbers as $v_{k}=-E_{k}\left(0\right),\thinspace\thinspace k\ge1.$

\subsection{A proof using generating functions}

We give here another proof based on generating functions: consider
the bivariate exponential generating function of the left-hand side
of (\ref{eq:pain2})
\[
\sum_{m,n\ge0}\left[p^{n},q^{m}\right]\frac{u^{n}}{n!}\frac{v^{m}}{m!}=\left[e^{up},e^{vq}\right].
\]
For the right-hand side, we have the bivariate generating function
\begin{eqnarray*}
\sum_{m,n\ge0}\sum_{k=1}^{\min\left(m,n\right)}c^{k}k!v_{k}\binom{n}{k}\binom{m}{k}\left\{ p^{n-k},q^{m-k}\right\} \frac{u^{n}}{n!}\frac{v^{m}}{m!} & = & \sum_{m,n\ge0}\sum_{k=1}^{\min\left(m,n\right)}\frac{\left(cuv\right)^{k}}{k!}v_{k}\left\{ \frac{\left(up\right)^{n-k}}{\left(n-k\right)!},\frac{\left(vq\right)^{m-k}}{\left(m-k\right)!}\right\} \\
=\left(\sum_{k=1}^{+\infty}\frac{\left(cuv\right)^{k}}{k!}v_{k}\right)\sum_{m,n\ge0}\left\{ \frac{\left(up\right)^{n}}{n!},\frac{\left(vq\right)^{m}}{m!}\right\}  & = & \left(\sum_{k=1}^{+\infty}\frac{\left(cuv\right)^{k}}{k!}v_{k}\right)\left\{ e^{up},e^{vq}\right\} 
\end{eqnarray*}
and we deduce, denoting $z=cuv,$
\begin{equation}
\left[e^{up},e^{vq}\right]=\left(\sum_{k=1}^{+\infty}\frac{z^{k}}{k!}v_{k}\right)\left\{ e^{up},e^{vq}\right\} .\label{eq:generating}
\end{equation}
But since $\left[p,q\right]=c,$ using McCoy's identity \cite[Formula (14)]{McCoy}
\begin{equation}
\left[f\left(p\right),g\left(q\right)\right]=-\sum_{k\ge1}\frac{\left(-c\right)^{k}}{k!}f^{\left(k\right)}\left(p\right)g^{\left(k\right)}\left(q\right)\label{eq:McCoy}
\end{equation}
with $f\left(p\right)=\exp\left(up\right)$ and $f\left(q\right)=\exp\left(vq\right)$,
we deduce
\begin{equation}
\left[e^{up},e^{vq}\right]=\left(1-e^{-z}\right)e^{up}e^{vq}\label{eq:eupevqcommutator}
\end{equation}
and, since $-\left\{ e^{up},e^{vq}\right\} =\left[e^{up},e^{vq}\right]-2e^{up}e^{vq},$
we deduce
\begin{equation}
\left\{ e^{up},e^{vq}\right\} =\left(1+e^{-z}\right)e^{up}e^{vq}.\label{eq:eupevqanticommutator}
\end{equation}

In fact, these two identities can be obtained without using McCoy's
identity (\ref{eq:McCoy}) as a consequence of the celebrated Baker-Campbell-Hausdorff
formula that tells us a more precise result, namely that
\[
e^{p}e^{q}=e^{p+q+\frac{1}{2}\left[p,q\right]}=e^{p+q+\frac{c}{2}}
\]
and
\[
e^{q}e^{p}=e^{q+p-\frac{c}{2}};
\]
from these two formulas, we deduce
\[
e^{p+q}=e^{-\frac{c}{2}}e^{p}e^{q}=e^{\frac{c}{2}}e^{q}e^{p}
\]
and then (\ref{eq:eupevqcommutator}) and (\ref{eq:eupevqanticommutator}).

Hence the commutators and anti-commutators are related as
\begin{eqnarray}
\left[e^{up},e^{vq}\right] & = & \frac{1-e^{-z}}{1+e^{-z}}\left\{ e^{up},e^{vq}\right\} =\frac{e^{z}-1}{e^{z}+1}\left\{ e^{up},e^{vq}\right\} .\label{eq:comm anticomm}
\end{eqnarray}
Identifying with (\ref{eq:generating}) gives
\[
\sum_{k=1}^{+\infty}\frac{z^{k}}{k!}v_{k}=\frac{e^{z}-1}{e^{z}+1}
\]
or
\[
\sum_{k=0}^{+\infty}\frac{z^{k}}{k!}v_{k}=1+\frac{e^{z}-1}{e^{z}+1}=\frac{2e^{z}}{1+e^{z}}=\sum_{k=0}^{+\infty}\frac{E_{k}\left(1\right)}{k!}z^{k}
\]
so that $v_{k}=E_{k}\left(1\right).$ Using (\ref{eq:Eulermean}),
we deduce, for all $k\ge1,$
\[
v_{k}=E_{k}\left(1\right)=-E_{k}\left(0\right).
\]

\subsection{The reciprocal identity}

One of the interesting features of this generating function method
is that it allows us to invert the identity (\ref{eq:pain}), giving
an expression of the anti-commutators in terms of the commutators.
First recall that the Bernoulli numbers are defined by the exponential
generating function
\[
\sum_{n\ge0}\frac{B_{n}}{n!}z^{n}=\frac{z}{e^{z}-1},\thinspace\thinspace\vert z\vert<2\pi.
\]
We can now prove the following result.
\begin{thm}
If $p$ and $q$ satisfy $\left[p,q\right]=c$ then for $m,n\ge0,$
\begin{equation}
\left\{ \frac{p^{n}}{n!},\frac{q^{m}}{m!}\right\} =\frac{2}{c}\left[\frac{p^{n+1}}{\left(n+1\right)!},\frac{q^{m+1}}{\left(m+1\right)!}\right]+2\sum_{k\ge1}^{\min\left(m,n\right)}\frac{c^{k}}{k!}\frac{B_{k+1}}{k+1}\left[\frac{p^{n-k}}{\left(n-k\right)!},\frac{q^{m-k}}{\left(m-k\right)!}\right]\label{eq:pain3}
\end{equation}
where $B_{n}$ is the $n-$th Bernoulli number.\end{thm}
\begin{proof}
First transform the identity
\[
\left[e^{up},e^{vq}\right]=\frac{e^{z}-1}{e^{z}+1}\left\{ e^{up},e^{vq}\right\} 
\]
into the equivalent
\begin{eqnarray*}
\left\{ e^{up},e^{vq}\right\}  & = & \frac{e^{cuv}+1}{e^{cuv}-1}\left[e^{up},e^{vq}\right]=\left(1+\frac{2}{e^{z}-1}\right)\left[e^{up},e^{vq}\right]\\
 & = & \left(1+\frac{2}{z}\frac{z}{e^{z}-1}\right)\left[e^{up},e^{vq}\right]=\left(1+\frac{2}{z}\sum_{k\ge0}\frac{z^{k}}{k!}B_{k}\right)\left[e^{up},e^{vq}\right].
\end{eqnarray*}
Since $B_{0}=1$ and $B_{1}=-\frac{1}{2},$ this gives
\[
\left\{ e^{up},e^{vq}\right\} =\left(1+\frac{2}{z}\left(1-\frac{1}{2}z+\sum_{k\ge2}\frac{z^{k}}{k!}B_{k}\right)\right)\left[e^{up},e^{vq}\right]
\]
which simplifies to
\[
\left\{ e^{up},e^{vq}\right\} =2\left(\frac{1}{z}+\sum_{k\ge1}\frac{z^{k}}{k!}\frac{B_{k+1}}{k+1}\right)\left[e^{up},e^{vq}\right].
\]
Expanding the exponentials in the commutator and anti-commutator and
identifying the corresponding powers of $u$ and $v$ with $z=cuv$
gives the desired result.
\end{proof}
As remarked by J.-C. Pain in \cite{Pain}, using the fact that 
\[
E_{k}\left(0\right)=-2\left(2^{k+1}-1\right)\frac{B_{k+1}}{k+1},
\]
the original identity (\ref{eq:pain}) can be rewritten as
\[
\left[\frac{p^{n}}{n!},\frac{q^{m}}{m!}\right]=2\sum_{k=1}^{\min\left(m,n\right)}\frac{c^{k}}{k!}\left(2^{k+1}-1\right)\frac{B_{k+1}}{k+1}\left\{ \frac{p^{n-k}}{\left(n-k\right)!},\frac{q^{m-k}}{\left(m-k\right)!}\right\} 
\]
which now shows more resemblance with its reciprocal version (\ref{eq:pain3}).

\subsection{An analytic approach}

We look now for another proof of the identity (\ref{eq:pain}) using
an analytic approach: the operators $p$ and $q$ are realized as
\[
p=cf'\left(x\right),\thinspace\thinspace qf\left(x\right)=xf\left(x\right)
\]
so that they satisfy the condition
\[
\left[p,q\right]=c,
\]
and we look at the action of the commutators $\left[p^{n},q^{m}\right]$
and anti-commutators $\left\{ p^{n},q^{m}\right\} $ on the monomials
$f\left(x\right)=x^{l}$ where we assume that $l\ge n.$ An easy computation
gives \footnote{\,the usual convention for binomial coefficients is used here: $\binom{n}{k}=0$
if $k>n$.} 
\[
\left[\frac{p^{n}}{n!},\frac{q^{m}}{m!}\right]x^{l}=\frac{c^{n}}{m!}\left(\binom{m+l}{n}-\binom{l}{n}\right)x^{l-n+m}
\]
and
\[
\left\{ \frac{p^{n-k}}{\left(n-k\right)!},\frac{q^{m-k}}{\left(m-k\right)!}\right\} x^{l}=\frac{c^{n-k}}{\left(m-k\right)!}\left(\binom{m-k+l}{n-k}+\binom{l}{n-k}\right)x^{l-n+m}.
\]
Hence we need to prove that
\[
\binom{m+l}{n}-\binom{l}{n}=-\sum_{k=1}^{\min\left(m,n\right)}E_{k}\left(0\right)\binom{m}{k}\left(\binom{m-k+l}{n-k}+\binom{l}{n-k}\right)
\]
or equivalently, since the Euler polynomials satisfy $E_{k}\left(1\right)=-E_{k}\left(0\right)$
for $k\ge1,$ 
\begin{equation}
\sum_{k=0}^{\min\left(m,n\right)}E_{k}\left(0\right)\binom{m}{k}\binom{m-k+l}{n-k}=\sum_{k=0}^{\min\left(m,n\right)}E_{k}\left(1\right)\binom{m}{k}\binom{l}{n-k}.\label{eq:pain analytic}
\end{equation}
We will prove the following more general result.
\begin{thm}
The following identity holds
\begin{equation}
\sum_{k=0}^{\min\left(m,n\right)}z^{k}\binom{m}{k}\binom{m-k+l}{n-k}=\sum_{k=0}^{\min\left(m,n\right)}\left(z+1\right)^{k}\binom{m}{k}\binom{l}{n-k}\label{eq:general-2}
\end{equation}
As a consequence, the Euler polynomials satisfy the identity
\begin{equation}
\sum_{k=0}^{\min\left(m,n\right)}E_{k}\left(z\right)\binom{m}{k}\binom{m-k+l}{n-k}=\sum_{k=0}^{\min\left(m,n\right)}E_{k}\left(z+1\right)\binom{m}{k}\binom{l}{n-k}.\label{eq:general-1-1}
\end{equation}
The particular case $z=0$ gives the desired identity (\ref{eq:pain analytic}).\end{thm}
\begin{proof}
We first proceed to prove the first identity (\ref{eq:general-2}).
Defining 
\[
P\left(z\right)=\sum_{k=0}^{\min\left(m,n\right)}z^{k}\binom{m}{k}\binom{l}{n-k},\thinspace\thinspace Q\left(z\right)=\sum_{k=0}^{\min\left(m,n\right)}z^{k}\binom{m}{k}\binom{m-k+l}{n-k},
\]
we thus need to show that $P\left(z+1\right)=Q\left(z\right).$ Expanding
each $\left(z+1\right)^{k}$ in the expression for $P\left(z+1\right)$
gives
\begin{eqnarray*}
\sum_{k=0}^{\min\left(m,n\right)}\left(z+1\right)^{k}\binom{m}{k}\binom{l}{n-k} & = & \sum_{k=0}^{\min\left(m,n\right)}\sum_{j=0}^{k}\binom{k}{j}z^{j}\binom{m}{k}\binom{l}{n-k}\\
 & = & \sum_{j=0}^{\min\left(m,n\right)}\binom{m}{j}\sum_{k=0}^{m-j}\binom{m-j}{k}\binom{l}{n-k-j}z^{j}.
\end{eqnarray*}
Considering the coefficient of $z^{j}$ in this last expression, we
see that the identity to be proved is
\[
\binom{m-j+l}{n-j}=\sum_{k=0}^{m-j}\binom{m-j}{k}\binom{l}{n-k-j}.
\]
This is known as the {\it Chu-Vandermonde identity}, and is a direct consequence
of the identity
\[
\left(1+t\right)^{m-j+l}=\left(1+t\right)^{m-j}\left(1+t\right)^{l},
\]
by considering the coefficient of $t^{n-j}$ in each side.

Another proof based on hypergeometric functions is as follows: identify
each of the polynomials $P$ and $Q$ as a Gau\ss\,\,hypergeometric
function
\[
P\left(z\right)=\binom{l}{n}\thinspace_{2}F_{1}\left(\begin{array}{c}
-m,-n\\
l+1-n
\end{array};z\right),\thinspace\thinspace\thinspace Q\left(z\right)=\binom{m+l}{n}\thinspace_{2}F_{1}\left(\begin{array}{c}
-m,-n\\
-l-m
\end{array};-z\right).
\]
The desired identity $P\left(z+1\right)=Q\left(z\right)$ is a simple
consequence of the linear transformation formula for the Gau\ss\,\,hypergeometric
function that appears in \cite{Abramowitz} as identity 15.3.6,
\begin{eqnarray*}
\thinspace_{2}F_{1}\left(\begin{array}{c}
a,b\\
c
\end{array};z\right) & = & \frac{\Gamma\left(c\right)\Gamma\left(c-a-b\right)}{\Gamma\left(c-a\right)\Gamma\left(c-b\right)}\thinspace_{2}F_{1}\left(\begin{array}{c}
a,b\\
a+b-c+1
\end{array};1-z\right)\\
 & + & \left(1-z\right)^{c-a-b}\frac{\Gamma\left(c\right)\Gamma\left(a+b-c\right)}{\Gamma\left(a\right)\Gamma\left(b\right)}\thinspace_{2}F_{1}\left(\begin{array}{c}
c-a,c-b\\
c-a-b+1
\end{array};1-z\right)
\end{eqnarray*}
with the choice $a=-m,\thinspace b=-n,\thinspace c=l+1-n;$ note that
the second term in the right-hand side vanishes since we assumed $-n+l+1>0$
and since $\Gamma\left(-n\right)=+\infty.$

The last step of the proof is to transform the general identity (\ref{eq:general-2})
into identity (\ref{eq:general-1-1}) for Euler polynomials; this
is obtained by using the integral representation (\ref{eq:Euler integral})
of Euler polynomials. Replacing $z^{k}$ by $\left(\imath u-\frac{1}{2}+z\right)^{k}$
on each side of (\ref{eq:general-2}) and integrating over $u$ against
the $sech\left(\pi u\right)$ function gives the desired result.
\end{proof}

\subsection{A generalization}

Exploiting the linearity of identities (\ref{eq:pain}) and (\ref{eq:pain3})
allows to extend them to more general functions. We explicit this
generalization in the case of two functions $f\left(p\right)$ and
$g\left(p\right)$ that we assume regular enough to have an integral
representation as Fourier transforms:
\begin{equation}
f\left(p\right)=\int_{\mathbb{R}}\tilde{f}\left(u\right)e^{\imath up}du,\thinspace\thinspace g\left(q\right)=\int_{\mathbb{R}}\tilde{g}\left(v\right)e^{\imath vq}dq.\label{eq:Laplace}
\end{equation}

\begin{thm}
Assume that the functions $f$ and $g$ have the Fourier integral
representation (\ref{eq:Laplace}) for the operators $p$ and $q$
such that $\left[p,q\right]=c.$ Denote moreover 
\[
F\left(p\right)=\int_{\mathbb{R}}\frac{\tilde{f}\left(u\right)}{\imath u}e^{\imath up}du
\]
an antiderivative of $f$ and accordingly $G\left(p\right)$ for an antiderivative
of $g.$ 

Then we have
\[
\left[f\left(p\right),g\left(q\right)\right]=-\sum_{k\ge1}E_{k}\left(0\right)\frac{c^{k}}{k!}\left\{ f^{\left(k\right)}\left(p\right),g^{\left(k\right)}\left(q\right)\right\} 
\]
and
\[
\left\{ f\left(p\right),g\left(q\right)\right\} =\frac{2}{c}\left[F\left(p\right),G\left(q\right)\right]+2\sum_{k\ge1}\frac{B_{k+1}}{k+1}\frac{c^{k}}{k!}\left[f^{\left(k\right)}\left(p\right),g^{\left(k\right)}\left(q\right)\right].
\]
Additionally to McCoy's identity \cite[Formula (14)]{McCoy}
\[
\left[f\left(p\right),g\left(q\right)\right]=-\sum_{k\ge1}\frac{\left(-c\right)^{k}}{k!}f^{\left(k\right)}\left(p\right)g^{\left(k\right)}\left(q\right),
\]
we have the obvious
\[
\left\{ f\left(p\right),g\left(q\right)\right\} =2f\left(p\right)g\left(q\right)+\sum_{k\ge1}\frac{\left(-c\right)^{k}}{k!}f^{\left(k\right)}\left(p\right)g^{\left(k\right)}\left(q\right)
\]
and the two last cases are
\[
f\left(p\right)g\left(q\right)=\frac{1}{c}\left[F\left(p\right),G\left(q\right)\right]-\sum_{k\ge0}\frac{B_{k+1}}{k+1}\frac{\left(-c\right)^{k}}{k!}\left[f^{\left(k\right)}\left(p\right),g^{\left(k\right)}\left(q\right)\right]
\]
and
\[
f\left(p\right)g\left(q\right)=\frac{1}{2}\sum_{k\ge0}\frac{E_{k}\left(0\right)}{k!}\left(-c\right)^{k}\left\{ f^{\left(k\right)}\left(p\right),g^{\left(k\right)}\left(q\right)\right\} .
\]
\end{thm}
\begin{proof}
We give the proof of the first identity only, the other proofs follow
the same pattern. 

The commutator $\left[f\left(p\right),g\left(q\right)\right]$ is
computed from the integral representations (\ref{eq:Laplace}) as
\[
\iint_{\mathbb{R}^{2}}\tilde{f}\left(u\right)\tilde{g}\left(v\right)\left[e^{\imath up},e^{\imath vq}\right]dudv=\iint_{\mathbb{R}^{2}}\tilde{f}\left(u\right)\tilde{g}\left(v\right)\frac{e^{-z}-1}{e^{-z}+1}\left\{ e^{\imath up},e^{\imath vq}\right\} dudv
\]
where we have used (\ref{eq:comm anticomm}) replacing $\left(u,v\right)$
by $\left(\imath u,\imath v\right)$ and remarking that $z=cuv$ is
changed to $-z$. Now writing
\[
\frac{e^{-z}-1}{e^{-z}+1}=\frac{1-e^{z}}{1+e^{z}}=1-\frac{2e^{z}}{e^{z}+1}=1-\sum_{k\ge0}\frac{E_{k}\left(1\right)}{k!}z^{k},
\]
we deduce
\begin{eqnarray*}
\iint_{\mathbb{R}^{2}}\tilde{f}\left(u\right)\tilde{g}\left(v\right)\left[e^{\imath up},e^{\imath vq}\right]dudv & = & \iint_{\mathbb{R}^{2}}\tilde{f}\left(u\right)\tilde{g}\left(v\right)\left(1-\sum\frac{E_{k}\left(1\right)}{k!}z^{k}\right)\left\{ e^{\imath up},e^{\imath vq}\right\} dudv\\
 & = & \iint_{\mathbb{R}^{2}}\tilde{f}\left(u\right)\tilde{g}\left(v\right)\left\{ e^{\imath up},e^{\imath vq}\right\} dudv\\
 & - & \sum_{k\ge0}\frac{E_{k}\left(1\right)c^{k}}{k!}\iint_{\mathbb{R}^{2}}u^{k}\tilde{f}\left(u\right)v^{k}\tilde{g}\left(v\right)\left\{ e^{\imath up},e^{\imath vq}\right\} dudv.
\end{eqnarray*}
The first right-hand side term is simply the anti-commutator $\left\{ f\left(p\right),g\left(q\right)\right\} .$
Each integral in the second term is identified as
\[
\int_{\mathbb{R}}u^{k}\tilde{f}\left(u\right)e^{\imath up}du=\left(-\imath\right)^{k}f^{\left(k\right)}\left(p\right),
\]
and the same for $g.$ We deduce the second right-hand side term as
\[
\sum_{k\ge0}\frac{E_{k}\left(1\right)\left(-c\right)^{k}}{k!}\left\{ f^{\left(k\right)}\left(p\right),g^{\left(k\right)}\left(q\right)\right\} .
\]
Hence the right-hand side reads
\[
\left\{ f\left(p\right),g\left(q\right)\right\} -\sum_{k\ge0}\frac{E_{k}\left(1\right)\left(-c\right)^{k}}{k!}\left\{ f^{\left(k\right)}\left(p\right),g^{\left(k\right)}\left(q\right)\right\} =-\sum_{k\ge1}\frac{E_{k}\left(1\right)\left(-c\right)^{k}}{k!}\left\{ f^{\left(k\right)}\left(p\right),g^{\left(k\right)}\left(q\right)\right\} .
\]
The final result is obtained by remarking that $E_{k}\left(1\right)=\left(-1\right)^{k}E_{k}\left(0\right),\thinspace\thinspace k\ge0.$
\end{proof}

\section{non-Hermitian Hamiltonian systems}

\subsection{Figuieira de Morisson and Fring's results}

Another context in which the Euler numbers appear naturally is the
study of non-Hermitian Hamiltonian systems: in \cite{Figueira}, Figuieira
de Morisson and Fring consider a system ruled by an Hamiltonian $H$
that is not Hermitian but PT-symmetric, hence pseudo-Hermitian \cite{Mostafazadeh};
it is assumed that $H$ is similar to the Hermitian Hamiltonian $h,$
and that the similarity transformation can be realized under the form
\begin{equation}
h=e^{\frac{q}{2}}He^{-\frac{q}{2}}\label{eq:hq}
\end{equation}
for some Hermitian operator $q$. Writing the non-Hermitian Hamiltonian
\[
H=h_{0}+\imath h_{1}
\]
where $h_{0}^{\dagger}=h_{0}$ and $h_{1}^{\dagger}=h_{1}$, the pseudo-Hermitianity
condition
\[
H^{\dagger}=e^{q}He^{-q}
\]
is expressed as
\[
e^{q}\left(h_{0}-\imath h_{1}\right)e^{-q}=h_{0}+\imath h_{1}
\]
or equivalently as 
\begin{equation}
h_{0}-e^{q}h_{0}e^{-q}=\imath\left(h_{1}+e^{q}h_{1}e^{-q}\right).\label{eq:h0h1}
\end{equation}
This equation is solved in $h_{1}$ as a function of $h_{0}$ in \cite{Figueira},
looking for coefficients $\kappa_{n}$ such that
\begin{equation}
h_{1}=\imath\sum_{n=1}\frac{\kappa_{n}}{n!}\left[h_{0},q\right]_{n}.\label{eq:kappa}
\end{equation}
In order to comply with the notations in \cite{Figueira}, note that
in this section, the nested commutator in (\ref{eq:kappa}) is now
iterated from the left:
\[
\left[h_{0},q\right]_{n}=\left[q,\left[q,\dots,\left[q,h_{0}\right]\right]\right].
\]
It is found in \cite{Figueira} that all even-index $\kappa_{2n}$
numbers vanish and that
\[
\kappa_{1}=\frac{1}{2},\thinspace\thinspace\kappa_{3}=-\frac{1}{4},\thinspace\thinspace\kappa_{5}=\frac{1}{2},\thinspace\thinspace\kappa_{7}=-\frac{17}{8},\thinspace\thinspace\kappa_{9}=\frac{31}{2}...
\]
(note the sign error on $\kappa_{9}$ in \cite{Figueira}). Then $h$
is computed using (\ref{eq:hq}) and it is found that
\begin{equation}
h=\sum_{n\ge0}\frac{\lambda_{n}}{2^{n}n!}\left[h_{0},q\right]_{n}\label{eq:lambda}
\end{equation}
where the coefficients $\lambda_{n}$ are related to the $\kappa_{n}$
by
\[
\lambda_{n}=1-\sum_{m=0}^{n}2^{m}\binom{n}{m}\kappa_{m}.
\]
These coefficients are identified as $\lambda_{2n+1}=0$ and
\[
\lambda_{0}=1,\thinspace\thinspace\lambda_{2}=-1,\thinspace\thinspace\lambda_{4}=5,\thinspace\thinspace\lambda_{6}=-61...
\]
from where it is concluded (without proof) that
\[
\lambda_{2n}=2^{n}E_{n}\left(\frac{1}{2}\right).
\]

\subsection{A proof and why Euler numbers appear here}

In order to prove these results, we'll use a new method based on symbolic
computation in the spirit of the classical umbral calculus as introduced
in \cite{roman}, see also \cite{Gessel}. We define an Euler symbol
$\mathcal{E}$ such that
\[
\mathcal{E}^{n}=E_{n}\left(0\right);
\]
the principle of umbral calculus is to perform the computations replacing
the sequence of numbers $E_{n}$ by the powers $\mathcal{E}^{n}$
and to obtain the final result by the reverse substitution of each
$\mathcal{E}^{n}$ by its counterpart $E_{n}.$

For example,
\[
e^{z\mathcal{E}}=\sum_{n\ge0}\frac{\mathcal{E}^{n}}{n!}z^{n}=\sum_{n\ge0}\frac{E_{n}\left(0\right)}{n!}z^{n}=\frac{2}{e^{z}+1}
\]
and more generally
\begin{equation}
e^{z\left(\mathcal{E}+x\right)}=\sum_{n\ge0}\frac{\left(\mathcal{E}+x\right)^{n}}{n!}z^{n}=\sum_{n\ge0}\frac{E_{n}\left(x\right)}{n!}z^{n}=\frac{2e^{zx}}{e^{z}+1}\label{eq:Eulergf}
\end{equation}
so that $\mathcal{E}+x$ is the symbol for the Euler polynomial, in
the sense that
\[
\left(\mathcal{E}+x\right)^{n}=E_{n}\left(x\right).
\]
The fundamental property of Euler symbols that explains their appearance
in this context is the following extension of the identity (\ref{eq:Eulermean}).
\begin{lem}
For any analytic function $f,$
\begin{equation}
f\left(z+\mathcal{E}\right)+f\left(z+\mathcal{E}+1\right)=2f\left(z\right);\label{eq:Euler fundamental}
\end{equation}
\end{lem}
\begin{proof}
By linearity, this result needs only be proved for monomials $f\left(z\right)=z^{n}:$
it then reduces to identity (\ref{eq:Eulermean}).
\end{proof}
We are now in position to prove the following results.
\begin{thm}
The equation (\ref{eq:h0h1}) has a solution
\begin{equation}
h_{1}=\imath h_{0}-\imath\sum_{n\ge0}\frac{E_{n}\left(0\right)}{n!}\left[q,h_{0}\right]_{n}\label{eq:h1}
\end{equation}
The corresponding value of the Hamiltonian $h$ is
\begin{equation}
h=\sum_{n\ge0}\frac{E_{n}\left(\frac{1}{2}\right)}{n!}\left[q,h_{0}\right]_{n}.\label{eq:h}
\end{equation}
Hence the coefficients $\lambda_{n}$ in (\ref{eq:lambda}) and $\kappa_{n}$
in (\ref{eq:kappa}) are respectively equal to
\[
\kappa_{n}=\begin{cases}
0, & n=0\\
-E_{n}\left(0\right), & n>0
\end{cases}
\]
and
\[
\lambda_{n}=2^{n}E_{n}\left(\frac{1}{2}\right),\thinspace\thinspace n\ge0.
\]
\end{thm}
\begin{proof}
Rewrite first the right-hand side of (\ref{eq:h1}) as
\[
h_{1}=\imath h_{0}-\imath e^{q\mathcal{E}}h_{0}e^{-q\mathcal{E}};
\]
deduce
\begin{eqnarray*}
h_{1}+e^{q}h_{1}e^{-q} & = & \imath h_{0}-\imath e^{q\mathcal{E}}h_{0}e^{-q\mathcal{E}}+e^{q}\left(\imath h_{0}-\imath e^{q\mathcal{E}}h_{0}e^{-q\mathcal{E}}\right)e^{-q}\\
 & = & \imath\left(h_{0}+e^{q}h_{0}e^{-q}\right)-\imath\left(e^{q\mathcal{E}}h_{0}e^{-q\mathcal{E}}+e^{q\left(\mathcal{E}+1\right)}h_{0}e^{-q\left(\mathcal{E}+1\right)}\right).
\end{eqnarray*}
Applying (\ref{eq:Euler fundamental}) with $z=0$ and the function
$f\left(\mathcal{E}\right)=e^{q\mathcal{E}}h_{0}e^{-q\mathcal{E}}$
reduces the second term in the right-hand side to $2h_{0}$ so that
\[
h_{1}+e^{q}h_{1}e^{-q}=\imath\left(-h_{0}+e^{q}h_{0}e^{-q}\right)
\]
and $h_{1}$ is a solution to (\ref{eq:h0h1}). The corresponding
value of $h$ is
\begin{eqnarray*}
h & = & e^{\frac{q}{2}}\left(h_{0}+\imath h_{1}\right)e^{-\frac{q}{2}}=e^{\frac{q}{2}}\left(h_{0}-h_{0}+e^{q\mathcal{E}}h_{0}e^{-q\mathcal{E}}\right)e^{-\frac{q}{2}}\\
 & = & e^{q\left(\mathcal{E}+\frac{1}{2}\right)}h_{0}e^{-q\left(\mathcal{E}+\frac{1}{2}\right)}.
\end{eqnarray*}
Expanding the exponentials in the formula gives the result (\ref{eq:h}).
\end{proof}

\section{Conclusion}

We have studied three identities on commutators and anti-commutators
in which the Euler and Bernoulli numbers play a fundamental role.
An open problem and desirable result at this point would be to link
these identities with the combinatorial meaning of the Euler numbers:
see \cite{Hodges} for a possible application.

\end{document}